\documentclass[prd,tightenlines,nofootinbib,superscriptaddress,11pt]{revtex4}

\usepackage{amsfonts,amssymb,amsthm,bbm,amsmath}
\usepackage{color,psfrag}
\usepackage[dvips]{graphicx}

\newcommand{\C}{{\mathbb C}}
\newcommand{\N}{{\mathbb N}}
\newcommand{\R}{{\mathbb R}}

\newcommand{\cA}{{\mathcal A}}
\newcommand{\cB}{{\mathcal B}}

\newcommand{\cJ}{{\mathcal J}}

\newcommand{\cN}{{\mathcal N}}

\newcommand{\cC}{{\mathcal C}}

\newcommand{\SU}{\mathrm{SU}}
\newcommand{\SL}{\mathrm{SL}}
\newcommand{\GL}{\mathrm{GL}}

\newcommand{\U}{\mathrm{U}}

\newcommand{\be}{\begin{equation}}
\newcommand{\ee}{\end{equation}}
\newcommand{\beq}{\begin{eqnarray}}
\newcommand{\eeq}{\end{eqnarray}}
\newcommand{\bes}{\begin{eqnarray}}
\newcommand{\ees}{\end{eqnarray}}
\newcommand{\bea}{\begin{eqnarray}}
\newcommand{\eea}{\end{eqnarray}}
\newcommand{\nn}{\nonumber}

\newcommand{\Ref}[1]{(\ref{#1})}

\newcommand{\mat} [2] {\left ( \begin{array}{#1}#2\end{array} \right ) }

\newcommand{\su}{{\mathfrak{su}}}
\newcommand{\gl}{{\mathfrak{gl}}}
\newcommand{\so}{{\mathfrak{so}}}

\renewcommand{\u}{{\mathfrak{u}}}
\renewcommand{\sl}{{\mathfrak{sl}}}
\newcommand{\la}{\langle}
\newcommand{\ra}{\rangle}

\newcommand{\bra}[1]{\la {#1}|}
\newcommand{\ket}[1]{|{#1}\ra}

\newcommand{\w}{\wedge}

\newcommand{\tr}{\mathrm{Tr}}
\newcommand{\f}{\frac}
\newcommand{\tl}{\widetilde}

\newcommand{\eps}{\epsilon}
\newcommand{\g}{\gamma}

\newcommand{\bz}{\bar{z}}

\newcommand{\bF}{\bar{F}}
\newcommand{\bG}{\bar{G}}
\newcommand{\bU}{\bar{U}}
\newcommand{\vJ}{\vec{J}}
\newcommand{\vK}{\vec{K}}
\newcommand{\vL}{\vec{L}}

\newcommand{\va}{\vec{a}}
\newcommand{\vb}{\vec{b}}
\newcommand{\vc}{\vec{c}}
\newcommand{\vd}{\vec{d}}
\newcommand{\ve}{\vec{e}}
\newcommand{\vf}{\vec{f}}
\newcommand{\vv}{\vec{v}}
\newcommand{\vsigma}{\vec{\sigma}}
\newcommand{\id}{\mathbb{I}}

\def\bu{\bar{u}}
\def\bt{\bar{t}}
\def\bbeta{\bar{\beta}}
\def\tM{\widetilde{M}}

\def\arr{\rightarrow}
\def\pp{\partial}
\def\tTheta{\widetilde{\Theta}}

\newtheorem{theoreml}{Theorem}[]

\newtheorem{theoremp}{Theorem}[]
\newtheorem{lemma}[theoreml]{Lemma}
\newtheorem{prop}[theoremp]{Proposition}

\begin{document}

\title{Holomorphic Lorentzian Simplicity Constraints}

\author{{\bf Ma\"it\'e Dupuis}}
\affiliation{Laboratoire de Physique, ENS Lyon, CNRS-UMR 5672, 46 All\'ee d'Italie, Lyon 69007, France}
\affiliation{School of Physics, The University of Sydney,  New South Wales 2006, Australia}
\author{{\bf Laurent Freidel}}
\affiliation{Perimeter Institute, 31 Caroline St N, Waterloo ON, Canada N2L 2Y5}
\author{{\bf Etera R. Livine}}
\affiliation{Laboratoire de Physique, ENS Lyon, CNRS-UMR 5672, 46 All\'ee d'Italie, Lyon 69007, France}
\affiliation{Perimeter Institute, 31 Caroline St N, Waterloo ON, Canada N2L 2Y5}
\author{{\bf Simone Speziale}}
\affiliation{Centre de Physique Th\'eorique, CNRS-UMR 7332, Aix-Marseille Univ, Luminy Case 907, 13288 Marseille, France}

\date{\today}

\begin{abstract}
We develop a Hamiltonian representation of the $\sl(2,\C)$ algebra on a phase space consisting of $N$ copies of twistors, or bi-spinors. We identify a complete set of global invariants, and show that they generate a \emph{closed} algebra including $\gl(N,\C)$ as a subalgebra. Then, we define the linear and quadratic simplicity constraints which reduce the spinor variables to (framed) 3d spacelike polyhedra embedded in Minkowski spacetime.
Finally, we introduce a new version of the simplicity constraints which (i) are holomorphic
and (ii) Poisson-commute with each other, and show their equivalence to the linear and quadratic constraints.
\end{abstract}

\maketitle

\tableofcontents
\newpage


\section{Introduction}
The developments in spin foam models of the last few years \cite{LS,LS2,FK,EPR,EPRL,newlook} are based on a new way to implement the simplicity constraints which reduce BF theory to general relativity \cite{BFsimple,BCE,BCL}. It is now commonly accepted that they are to be imposed weakly, that is as conditions on the expectation values of operators, analogously to the Gupta--Bleuler prescription for the gauge fixing condition in quantum electrodynamics.
The results of this procedure are very promising: the new models have a boundary Hilbert space consistent with the one of loop quantum gravity, thus providing a well-defined dynamics for spin network states. Furthermore, such dynamics reduces to Regge calculus in the (large spin) semiclassical limit on a fixed triangulation \cite{Barrett,CF}, thus suggesting the correct recovery of general relativity at large scales.

These results put the general structure of the new models on a solid basis, and allow us to discuss
the details of the weak imposition, where some open questions remain. One concerns the fact that although the algebra of constraints is second class, it has a non-trivial center, the ``diagonal simplicity constraints'', given by a constraint on the Casimir operators of the Lorentz group. Unlike the rest of the constraints, the center is imposed strongly in the literature.\footnote{We also remind the reader that in the previous Barrett-Crane model \cite{BCE,BCL},
all constraints were imposed strongly, leading to inconsistencies.}
This somewhat uneven procedure for the implementation is a motivation for our present work, where we consider the possibility of treating all constraints on the same footing. This has been advocated in \cite{Alexandrov} and \cite{simplicity1,simplicity2}. Now, it is known that a naive weak implementation of all the constraints leads to a larger boundary Hilbert space with an extra, unphysical quantum number \cite{Ding}.
Here on the other hand, we are interested in the specific way of realizing the weak implementation of all the constraint through coherent states. These are peaked on the classical configurations but with a non-trivial and fixed spread, which should lead to no extra quantum numbers. The required coherent states are such that they satisfy (suitably chosen) half of the constraints \emph{strongly}. This half is selected as in the original Gupta--Bleuler formalism: there, the weak imposition is obtained starting from a splitting of the second class constraints into holomorphic and anti-holomorphic parts. These two parts are separately of first class. Then, the holomorphic part alone can be imposed strongly and this automatically induces a weak imposition of the whole set of constraints.

To realize this program, the spinorial and twisted geometry approach we have introduced in a series of papers \cite{twisted1,twisted2,un0,un1,un2,un3,spinor1,spinor2}
provides valuable tools, because it brings a natural complex structure into the spin foam formalism, which can be used to obtain the holomorphic/anti-holomorphic splitting of the constraints, to then impose the holomorphic part strongly.
This idea has been introduced by two of us in \cite{simplicity1,simplicity2} for the case of Euclidean signature. Here and in the follow-up work \cite{SL2Cquant}, we extend it to the Lorentzian case.
In the present paper, we develop the classical tools needed to the splitting.
We introduce a phase space carrying a complex structure and a Hamiltonian action of ($N$ copies of) $\SL(2,\C)$.
This space consists of $N$ twistors, or pairs of spinors, with canonical Poisson brackets.
We then study the space of global $\SL(2,\C)$ invariants. We identify a complete set, given by scalar products among the spinors, which generate a \emph{closed} Poisson algebra, with $\gl(N,\C)$ as a subalgebra.
Finally, we use the invariants to introduce a new version of the simplicity constraints. These are quadratic in the spinors, and satisfy the following key properties: $(i)$ they are holomorphic with respect to the complex structure, $(ii)$ they Poisson commute with each other, and $(iii)$ they are equivalent to the standard quadratic and linear simplicity constraints.

In the follow-up work \cite{SL2Cquant}, we will study the quantization of the new holomorphic constraints and its implications for loop quantum gravity and spin foam models.
The resulting spin foam amplitudes will be
slightly different from the EPRL-FK models \cite{LS,LS2,FK,EPR,EPRL}, due to the difference in dealing with the
diagonal simplicity constraints, but we nevertheless expect that they
will have the same asymptotics in the large spin limit.
That is to say, this type of details on the imposition of the constraints should not affect the semiclassical limit, but becomes relevant when looking at the quantum corrections.
Our phase space admits a nice interpretation in terms of twistors, which will play a role in the quantization \cite{SL2Cquant}.
This line of thoughts pushes forward the relation between loop gravity and twistors suggested in \cite{twisted2}.
More details on the classical phase space will appear in \cite{Johannes}. See also a work by Wieland coming out in parallel with this one \cite{Wolfgang}, where some aspects of our construction appear as well.

\section{Spinor Phase Space and $\U(N)$-Formalism for $\SU(2)$-Intertwiners}

\subsection{Spinor Phase Space and $\su(2)$-Algebra}
We begin with a review of some material appeared in \cite{twisted2}. Consider the classical vector space of spinors, $\ket{z}\in\C^2$, for which we use ket and bra notations,\footnote{Although everything we do is classical, Dirac's notation is convenient to avoid clogging the formulas with indices.}
$$
|z\ra=\mat{c}{z^0\\z^1}, \qquad \la z|=\mat{cc}{\bar{z}^0 &\bar{z}^1}.
$$
We also define the dual spinor using the complex conjugation map\footnote{The unitary operator $\eps$ satisfies $\eps^2=-\id$, $\eps^\dag=\eps^{-1}=-\eps$ and allows to switch the Pauli matrices to their complex conjugate,
$
\eps\sigma_a\eps\,=\, \overline{\sigma_a}.
$
This allows to check simply the following equalities, which are useful in manipulating spinors:
$$
[z|w]=\la w | z\ra,\qquad [z|w\ra = - [w|z\ra,
\qquad
[z|\vsigma|w]=-\la w |\vsigma| z\ra,\qquad [z|\vsigma|w\ra = [w|\vsigma|z\ra.
$$
}
on $\SU(2)$,
$$
|z] = \mat{c}{-\bar{z}^1\\\bar{z}^0} = \eps\,|\bz\ra,
\qquad
\eps=\mat{cc}{0 & -1 \\ 1&0}.
$$
The spinor defines a real, null vector $J_\mu = (J_0,\vJ\,)$ in Minkowski space via the familiar formula
\be\label{zz}
|z\ra\la z|= J_0 \, \id + \vJ(z)\cdot\vsigma,
\ee
where $\vsigma$ are the Pauli matrices, and
\be\label{defJ}
J_0= \f12\la z|z\ra = |\vJ|, \qquad \vJ(z)=\f12\la z|\vsigma|z\ra=-\f12[ z|\vsigma|z].
\ee
The vector $\vJ$ characterizes the spinor up to a phase, $\phi=\arg z^0 + \arg z^1$.

Next, we endow the space of spinors with canonical Poisson brackets,
\be\label{zPB}
\{z^A,\bz^B\}=-i\delta^{AB},
\ee
where $A,B=0,1$ are spinorial indices. With these brackets, the components of
$\vJ(z)$ generate an SU(2) action on $\ket{z}$,
\be\label{su2PB}
\{J_a(z),J_b(z)\}=\epsilon_{abc}J_c(z),
\ee
\be
\{\vJ(z),|z\ra\}\,=\,i\f\vsigma2\,|z\ra,
\qquad
\{\vJ(z),|z]\}\,=\,i\f\vsigma2\,|z],
\ee
The infinitesimal action can be easily exponentiated,
\be
e^{\{\vec{u}\cdot\vJ(z)\,,\,\cdot\}}\,|z\ra
\,=\,
e^{i\vec{u}\cdot\f{\vsigma}{2}}\,|z\ra,
\ee
and recognized as the action of an $\SU(2)$ group element in the fundamental representation.
The squared norm of the spinor coincides with the Casimir of $\su(2)$,
$\vJ^2(z)=\f14 \la z|z\ra^2$.
Including $J_0$, one recovers a $\u(2)$-algebra on $\C^2$.

The result is just a classical version of Schwinger representation of the angular momentum.
This phase space structure can be quantized, with quantization condition $\bra{z}z\ra\in\N$, giving the irreducible representations of $\SU(2)$ with arbitrary spin $j=\f12\bra{z}z\ra$.
The operators $J_a(z)$ become the $\SU(2)$ generators, and $\vJ^2(z)$ is the $\SU(2)$-invariant quadratic Casimir operator.

\subsection{Classical Geometry of Spinors and Polyhedra}

Let us now consider a collection of $N$ copies of the above phase space, with spinors $z_i$, $i=1 \ldots N$.
We are interested in the space of $\SU(2)$-invariant observables, defined by the condition
\be
\vJ\,\equiv\,\sum_i \vJ(z_i)
= \f12\sum_i  \la z_i | \vsigma | z_i\ra = 0.
\ee
For later purposes, we notice from \Ref{zz} that $\vJ=0$ amounts to requiring that the 2$\times$2 matrix $\sum_i | z_i\ra\la z_i |$ is proportional to the identity,
\be
\sum_i \vJ(z_i)
\,=\,
0
\quad\Leftrightarrow\quad
\sum_i | z_i\ra\la z_i | =\f12\sum_i\la z_i| z_i\ra\,\id\,.
\ee
The geometrical interpretation of $N$ spinors satisfying the closure constraint $\vJ=0$ is that of a ``framed'' polyhedron \cite{un1,un2}. The $N$ spinors $z_i$ define the $N$ vectors $\vJ_i\equiv\vJ(z_i)$, which satisfy the condition $\sum_i \vJ_i=0$. Thanks to Minkowski's theorem, see \cite{polyhedron}, these vectors determine a unique convex polyhedron embedded in $\R^3$ with $N$ faces, such that the $\vJ_i$ are the normal vectors to each face (with their norm giving the area of the corresponding face). Details on the reconstruction of the polyhedron can be found in \cite{polyhedron}. The ``framing" corresponds to the presence on each face of a further variable, the spinor phase $\phi_i$ uncaptured by the $\vJ_i$'s.
The phase defines a local reference frame attached to each face,
and carries the extrinsic curvature of twisted geometries \cite{twisted2}.

As shown in \cite{un1,un2,un3,spinor1}, a complete set of  $\SU(2)$-invariant observables can be defined in terms of scalar products among the spinors,
\be\label{EF}
E_{ij}\equiv \la z_i|z_j\ra,\qquad
F_{ij}\equiv [z_i |z_j\ra.
\ee
The matrix $E$ is Hermitian, $E_{ij}=\overline{E}_{ji}$, while the matrix $F$ is holomorphic in the spinor variables and anti-symmetric, $F_{ij}=-F_{ji}$.
It is straightforward to check that they commute with the $\su(2)$ generators,
\be
\{\vJ,E_{ij}\}=\{\vJ,F_{ij}\}=0.
\ee
Moreover, as noticed in \cite{un0,un1,un2}, their Poisson brackets generate a nice algebra. In particular, the $E_{ij}$ by themselves form a $\u(N)$ algebra,
\be
\{E_{ij},E_{kl}\}
\,=\,
-i\,\left(
\delta_{jk}E_{il}-\delta_{il}E_{kj}
\right),
\ee
with group action on the spinors given by
\be
e^{\{\sum_{j,k} \alpha_{jk} E_{jk},\cdot\}}\,|z_i\ra =
\sum_j(e^{i\alpha})_{ij}\,|z_j\ra.
\ee
Here $U\,\equiv\,e^{i\alpha}\,\in\U(N)$
and $\alpha$ is a Hermitian $N\times N$ matrix labeling the $\u(N)$ Lie algebra element.
Since the observables $E_{ij}$ are $\SU(2)$-invariant, the induced $\U(N)$ action preserves the closure constraints \cite{un1,un2}, \be
\sum_{i,j,k} U_{ij}|z_j\ra\la z_k| {\bU_{ik}}
\,=\,
\sum_j |z_j\ra\la z_j|\,,
\ee
thanks to the unitarity of $U$.

All together, the $E$'s, $F$'s and $\bF$'s form the following closed algebra,
\bes
\{E_{ij},F_{kl}\}&=&
-i\,\left(
\delta_{ik}F_{lj}-\delta_{il}F_{kj}
\right),\\\nn
\{E_{ij},\bF_{kl}\}&=&
-i\,\left(
\delta_{jl}\bF_{ki}-\delta_{jk}\bF_{li}
\right),\\
\{F_{ij},\bF_{kl}\}&=&
-i\,\left(
\delta_{ik} E_{lj} - \delta_{il} E_{kj} -\delta_{jk}E_{li}+\delta_{jl} E_{ki}
\right).\nn
\ees
Observe that the $F$'s are not independent, but satisfy the following Pl\"ucker relations,
\be\label{Plucker}
F_{ij}F_{kl}=F_{ik} F_{jl} -F_{il}F_{jk},
\ee
which express the basic recoupling identity for the tensor product of two copies of the fundamental representation of $\SU(2)$ \cite{spinor1}.

Finally, as used in \cite{un3,spinor1}, these matrices also satisfy quadratic constraints when we impose the closure constraints $\sum_i \vJ(z_i)=0$:
\be\label{EFconstraints}
E^2= \f{\tr E}2\,E,
\qquad
FE= \f{\tr E}2\,F,
\qquad
\bF F= -\f{\tr E}2\,E\,,
\ee
where $\tr E=\sum_iE_{ii}=\sum_i\la z_i|z_i\ra=2\sum_i |\vJ(z_i)|$ gives (twice) the total area of the corresponding  polyhedron.

\medskip

The geometrical interpretation of the $\SU(2)$ observables $E$ and $F$ is given by the expression of the standard scalar product between the normal vectors of each face of the polyhedron. Indeed, the scalar product $\vJ(z_i)\cdot \vJ(z_j)$ are also $\SU(2)$-invariant and be expressed in term of the $E$ and $F$ variables as follows,
\be
\vJ(z_i)\cdot \vJ(z_j)
\,=\,
\f12E_{ij}E_{ji}-\f14E_{ii}E_{jj}
\,=\,
-\f12 F_{ij}\bF_{ij}+\f14E_{ii}E_{jj}.
\ee

\medskip

At the quantum level, the $E_{ij}$ operators will still generate $\U(N)$ transformations acting on SU(2) intertwiner states leaving the total area $\tr E$ invariant. On the other hand, the $F_{ij}$ will act as annihilation operators decreasing the area while the $\bF_{ij}$ will become creation operators increasing the area.
See \cite{un0,un1,un2,un3} for more details.

\subsection{Probing the $\U(N)$ Structure}\label{SecSquash}

An interesting aspect of the $\U(N)$ action is that it can be used to obtain any polyhedron starting from a degenerate, ``squashed'' polyhedron with only two non-trivial faces.
This can be checked explicitly with the action defined above. Consider the following ``squashed'' configuration,
$$
|z_1\ra=\mat{c}{\lambda\\ 0},\qquad
|z_2\ra=|z_1]=\mat{c}{0\\ \lambda},\qquad
|z_{k\ge3}\ra=0,
\qquad
\lambda\in\R^+.
$$
Acting with an arbitrary $U\in\U(N)$, we get a collection of spinors determined by $\lambda$ and the matrix elements of $U$,
\be
|z_i\ra
\,=\,
\lambda\,\mat{c}{U_{i1}\\ U_{i2}}.
\ee
From this expression, one simply realizes that the closure constraint on the $z_i$'s is equivalent to the orthonormality of the two N-vectors $U_{i1}$ and $U_{i2}$:
\be
\sum_i |z_i\ra\la z_i|\propto\id
\quad\Leftrightarrow\quad
\sum_i |U_{i1}|^2-|U_{i2}|^2= \sum_i \bU_{i1}U_{i2}=\sum_i U_{i1}\bU_{i2}=0,
\ee
which is ensured by the unitarity of the matrix $U$. Thus acting with a $\U(N)$-transformation on the squashed configuration allows to get arbitrary collections of spinors satisfying the closure constraints.

Conversely, we can also see this from the point of view of the $E$ matrix. The quadratic constraint in \Ref{EFconstraints} implies that $E$ is diagonalizable, of rank 2 and with equal eigenvalues. Hence, there exists a unitary matrix $U\in\U(N)$ such that
\be
E=U
\,\mat{cc|c}{\lambda & & \\ &\lambda & \\ \hline && 0_{N-2}}\,
U^{-1},
\ee
where $\lambda$ is the unique eigenvalue of degeneracy 2. The diagonal matrix represents the squashed polyhedron with only two non-vanishing spinors.

At the quantum level, this means that the Hilbert space of intertwiners (at fixed total area) is  an irreducible representation of $\U(N)$ with highest weight given by a bivalent intertwiner \cite{un2}.

\subsection{Generalizing to $\SL(2,\C)$: the Basic Idea}

The purpose of this paper is to extend the spinor phase space and classical Schwinger representation from $\SU(2)$ to $\SL(2,\C)$.
As a first step, one easily notices that the single spinor phase space considered thus far also carries a representation of $\SL(2,\C)$. Indeed, introducing next to \Ref{defJ} the following boost generators,
\be\label{defK1}
\vK(z)=\f12 {\rm Re} [z|\vsigma|z\ra,
\ee
we obtain %
\be
\{J_a(z),K_b(z)\}=\{K_a(z),J_b(z)\}=\epsilon_{abc}K_c(z),\qquad
\{K_a(z),K_b(z)\}=-\epsilon_{abc}J_c(z),
\ee
which together with \Ref{su2PB} form an $\sl(2,\C)$ Lie algebra.\footnote{
Alternatively, one can take as boost generators the imaginary part of $[z|\vsigma|z\ra$,
that is
$
{\vec M}(z)=i/4\left(\la z|\vsigma|z]-[z|\vsigma|z\ra\right).
$
The generators $\vJ(z)$, $\vK(z)$, $\vec M(z)$ and $D\equiv\f12\la z | z\ra$ form a $\so(3,2)$ algebra. The interested reader will find the details in Appendix \ref{AppA}.
}
The rotations act as before, whereas the boosts mix holomorphic and anti-holomorphic components,\footnote{To see the mixing more precisely, we can look at a finite transformation on the bi-spinor
$Z\,\equiv\,(|z\ra,|z])$,
$$
e^{\{\vec{u}\cdot\vJ(z)\,,\,\cdot\}}\,\vartriangleright\,Z
\,=\,
\mat{cc}{e^{i\vec{u}\cdot\f{\vsigma}{2}}& 0 \\0&e^{i\vec{u}\cdot\f{\vsigma}{2}}}\,Z\,,
\qquad
e^{\{\vec{u}\cdot\vK(z)\,,\,\cdot\}}\,\vartriangleright\,Z
\,=\,
\mat{cc}{\cosh\f{\vec{u}\cdot\vsigma}{2}& i\sinh\f{\vec{u}\cdot\vsigma}{2} \\
-i\sinh\f{\vec{u}\cdot\vsigma}{2}&\cosh\f{\vec{u}\cdot\vsigma}{2}}\,Z.
$$
From this expression, it is clear that $\vK(z)$ generates a Bogoliubov transformation.
}
\be
\{\vK(z),|z\ra\}\,=\,i\f\vsigma2\,|z],
\qquad
\{\vK(z),|z]\}\,=\,-i\f\vsigma2\,|z\ra.
\ee

However, notice that both $\SL(2,\C)$-invariants vanish,
\be\label{C=0}
C_1 \equiv \vJ^2(z)-\vK^2(z) =0,\qquad
C_2 \equiv 2\vJ(z)\cdot\vK(z)=0.
\ee
This means that if we quantize this phase space with its algebra (turning as usual the classical quantities $\vJ(z),\vK(z)$ into $\SL(2,\C)$ generators acting on a suitable Hilbert space), we obtain only the (unitary) representation with fixed vanishing values for both Casimir operators.
On the other hand, we are interested in recovering all (unitary) representations of $\SL(2,\C)$ at the quantum level, thus this classical phase space is too limited for our purposes.
We need a classical phase space with both invariants taking arbitrary (real) values.

The simplest way to do so is to take the dual spinor $|z]$ independent from $|z\ra$, that is to work with a pair of spinors $(\ket{z}, |w])$. This is what we do in the next Section.
To motivate this approach, notice that \Ref{C=0} means that the  bivector
$J^{IJ}=(\vK,\vJ\,)\in\R^6\cong \w^2\,\R^{3,1}$
has zero norm. If we take two independent spinors, we have two independent null bivectors.
These can be used as a basis to generate a arbitrary bivectors in $\w^2\,\R^{3,1}$, thus arbitrary representations of $\SL(2,\C)$ (the interested reader will find a proof in Appendix \ref{AppB}).

\section{The Double-Spinor Phase Space for $\SL(2,\C)$}

\subsection{Representing $\sl(2,\C)$ with a Pair of Spinors}
We start with two copies of the previous phase space, a pair of spinors $(\ket{z},\ket{w}) \in \C^2\times \C^2$ with the same brackets for $z$ and $w$,
\be\label{zwPB}
\{z^A,\bz^B\}=-i\delta^{AB}, \qquad \{w^A,\bar w^B\}=-i\delta^{AB}.
\ee
We take $\vJ\equiv\vJ(z)+\vJ(w)$ as rotation generators and $\vK={\rm Re} [w|\vsigma|z\ra$
as  boost generators.\footnote{Again, alternatives are possible. Notice in fact that the bi-spinor space carries a representation of the larger group $O(4,2)$, the conformal group of Minkowski spacetime.}
It is straighforward to verify the correct Poisson brackets,
\beq
\{J_a,J_b\} = \epsilon_{abc}J_c \qquad
\{J_a,K_b\} = \epsilon_{abc}K_c \qquad
\{K_a,K_b\} = -\epsilon_{abc}J_c.
\eeq

A direct computation shows that the real quantities
\be\label{AB}
A\equiv \la z|z\ra-\la w|w\ra,\qquad
B\equiv \la z|w]+[ w|z\ra\,.
\ee
are invariant. The Casimirs evaluate to
\be
4(\vJ^2-\vK^2)=\,A^2-B^2,
\qquad
4\,\vJ\cdot\vK=\,AB.
\ee
As desired, they can take now arbitrary real values.

The disadvantage of this choice is that the Poisson brackets \Ref{zwPB} are not preserved by the boosts, which map $z$ into $w$ and viceversa. This is a consequence of the familiar fact that the Hermitian product over $\C^2$ is not preserved by $\SL(2,\C)$. To obtain invariant Poisson brackets, we introduce the left/right splitting
$\sl(2,\C)\cong \su(2)_L \otimes \su(2)_R$, with generators $\vJ^{L,R}=\vJ\pm i \vK$.
To this purpose, we introduce a left-handed spinor $\ket{u}$ and
a right-handed spinor $\ket{t}$,\footnote{The absence of a normalization factor $1/\sqrt{2}$ in the right-hand side will produce unconventional factors of 2 in the expressions for the canonical Poisson brackets and the SU(2) algebra. However, this choice will be convenient below when computing with the simplicity constraints.}
\be
|t\ra\equiv|z\ra-i|w],
\qquad
|u\ra\equiv|z\ra+i|w].
\ee
The Poisson brackets in these coordinates read
\be
\{t_A,\bu_B\}=\{u_A,\bt_B\}=-2i\delta_{AB},
\ee
with all the other vanishing.
These Poisson brackets are invariant under $\SL(2,\C)$.

In terms of the new spinors, the left-right $\su(2)$-generators $\vJ^{L,R}$ turn out to be
\be
\vJ^L=\f12\la t|\vsigma|u\ra,
\qquad
\vJ^R=\f12\la u|\vsigma|t\ra.
\ee
One can easily check the correct algebra,
\be
\{\vJ^L,\vJ^R\}=0,
\quad
\{J_a^L,J_b^L\}=2\eps_{abc}J_c^L,
\quad
\{J_a^R,J_b^R\}=2\eps_{abc}J_c^R,
\ee
and transformation properties,
\begin{align}
& \{\vJ^L,|u\ra\}\,=\,+i\,\vsigma \,|u\ra,
& \{\vJ^L,|t\ra\}\,=\,0,
\\
& \{\vJ^R,|t\ra\}\,=\,+i\vsigma \,|t\ra,
& \{\vJ^R,|u\ra\}\,=\,0,
\end{align}
For later purposes, we also give the exponentiated action,
\begin{align}
\label{JLaction}
& e^{i\vv\cdot\vJ^L}\vartriangleright |u\ra \,=\, e^{\{\vv\cdot\vJ^L,\cdot\}}\, |u\ra
\,=\, e^{+i\vv\cdot\vsigma}\, |u\ra,
& e^{i\vv\cdot\vJ^L}\vartriangleright |t\ra \,=\,|t\ra\,,
\\ \label{JRaction}
& e^{i\vv\cdot\vJ^R}\vartriangleright |t\ra \,=\, e^{\{\vv\cdot\vJ^R,\cdot\}}\, |t\ra
\,=\, e^{+i\vv\cdot\vsigma}\, |t\ra,
& e^{i\vv\cdot\vJ^R}\vartriangleright |u\ra \,=\,|u\ra\,,
\end{align}
where $\vv\in\C^3$ is an arbitrary complex parameter.\footnote{For completeness, here are the transformations on the bras,
\begin{align}\nn
& \{\vJ^L,\la u|\}\,=\,0, & \{\vJ^L,\la t|\}\,=\,-i\la t|\,\vsigma,
&& \{\vJ^R,\la t|\}\,=\,0, &&&\{\vJ^R,\la u|\}\,=\,-i\la u|\,\vsigma,\nn
\end{align}
and
\beq\nn
e^{i\vv\cdot\vJ^L}\vartriangleright \la u| \,=\, \la u|, \qquad
e^{i\vv\cdot\vJ^L}\vartriangleright \la t|\,=\,\la t|\,e^{-i\vv\cdot\vsigma},
\qquad 
e^{i\vv\cdot\vJ^R}\vartriangleright \la t| \,=\, \la t|, \qquad
e^{i\vv\cdot\vJ^R}\vartriangleright \la u| \,=\, \la u|\,e^{-i\vv\cdot\vsigma}.
\eeq
}
These formulas allow us to compute the action of arbitrary $\SL(2,\C)$ transformations on the spinors $t$ and $u$ and thus also on $z$ and $w$.

The $A$ and $B$ invariants can be packaged in the unique (complex) $\SL(2,\C)$ invariant
\be
\la u|t\ra \,=\, A-iB.
\ee
On the other hand, the scalar products $\la t|t\ra$, $\la u|u\ra$ and $[u|t\ra$ do not commute with $\vJ^L$ and $\vJ^R$ and are thus not $\SL(2,\C)$ invariant.
We remark also that the Casimirs of the self-dual and anti-self-dual $\SU(2)$'s are complex:
\be
(\vJ^L)^2=\f14\,\la t|u\ra^2
\,=\,
\f14(A+iB)^2,
\qquad
(\vJ^R)^2=\f14\,\la u|t\ra^2
\,=\,
\f14(A-iB)^2.
\ee

\medskip

Finally, the pair of spinors parametrizing the phase space can be viewed as a bi-spinor
\be
Z = (\ket{u}, \ket{t}).
\ee
This bi-spinor is the classical equivalent of a Dirac spinor, in the sense that
it is the direct sum of a left and a right-handed spinor, like in quantum theory.\footnote{
With this identification, the two invariants $A$ and $B$ become the familiar scalars $\bar \psi \psi$ and $\bar \psi \g^5 \psi$.}
The bi-spinor $Z$ also defines a twistor \cite{Penrose}. Following \cite{twisted2}, we are thus led to a description of the phase space in terms of twistors.
The relation with twistors will be developed further in the follow-up works \cite{Johannes} and \cite{SL2Cquant}.

\subsection{$\SL(2,\C)$-Invariants and the $\gl(N,\C)$-Structure}

We now consider a collection of $N$ pairs of spinors $(z_i,w_i)$, or equivalently $(t_i,u_i)$, with the index $i$ running from 1 to $N$, and the corresponding $N$ copies of the $\sl(2,\C)$ algebra with generators $\vJ_i,\vK_i$.
We are interested in the space of observables invariant under  global $\SL(2,\C)$ transformations, generated by $\vJ\equiv\sum_i\vJ_i$ and $\vK\equiv\sum_i\vK_i$.
A typical set of observables is provided by the scalar products,
\be\label{JJ}
\vJ^{L}_i\cdot\vJ^{L}_j, \qquad \vJ^{R}_i\cdot\vJ^{R}_j, \qquad \forall i,j.
\ee
It is not difficult to find also invariants in the form of scalar products among the spinors, building on the previous results \Ref{EF} and \Ref{AB}. These invariants consist of the generalization of the previous quantities $A$ and $B$, which are now $N\times N$ matrices,
\begin{subequations}\label{ABFG}\be
A_{ij}\equiv \la z_i|z_j\ra-\la w_j|w_i\ra,\qquad
B_{ij}\equiv \la z_i|w_j]+[ w_i|z_j\ra,
\ee
plus the new invariants
\be
F_{ij}\equiv [z_i|z_j\ra + \la w_i|w_j],\qquad
G_{ij}\equiv \la w_j|z_i\ra - \la w_i|z_j\ra.
\ee\end{subequations}
It is straightforward to check that all of the above expressions have vanishing Poisson brackets with both $\vJ$ and $\vK$.
The matrices $A_{ij}$ and $B_{ij}$ are Hermitian, while $F_{ij}$ and $G_{ij}$ are antisymmetric, purely holomorphic in $z$ and anti-holomorphic in $w$.

Unlike \Ref{JJ}, the invariants \Ref{ABFG} define a closed Poisson algebra. The $A$ and $B$ observables give a closed algebra by themselves,
\begin{subequations}\label{GLNC}\bes
\{A_{ij},A_{kl}\}
&=& -i(\delta_{jk}A_{il}-\delta_{il}A_{kj}), \\
\{A_{ij},B_{kl}\}
&=& -i(\delta_{jk}B_{il}-\delta_{il}B_{kj}), \nn\\
\{B_{ij},B_{kl}\}
&=& +i(\delta_{jk}A_{il}-\delta_{il}A_{kj}). \nn
\ees
This can be recognized to be a $\gl(N,\C)$ algebra, which is the complexification of the $\u(N)$ algebra found in the case of $\SU(2)$  invariants. The $A_{ij}$ operators generate simultaneous $\U(N)$ transformations on both spinors, $z\arr U\,z$ and $w\arr \bU\,w$, while the $B_{ij}$ operators mix the two sets.
The complete algebra of all observables $A,B,F,G$ has the following closed form,
\bes
\{F_{ij},\bF_{kl}\}
&=&
-i\,\left(
\delta_{ik} A_{lj} - \delta_{il} A_{kj} -\delta_{jk}A_{li}+\delta_{jl} A_{ki}
\right)
\,, \\
\{G_{ij},\bG_{kl}\}
&=&
+i\,\left(
\delta_{ik} A_{lj} - \delta_{il} A_{kj} -\delta_{jk}A_{li}+\delta_{jl} A_{ki}
\right)
\,, \nn\\
\{F_{ij},\bG_{kl}\}
&=&
-i\,\left(
\delta_{ik} B_{lj} - \delta_{il} B_{kj} -\delta_{jk}B_{li}+\delta_{jl} B_{ki}
\right)
\,, \nn\\
\{A_{ij},F_{kl}\}
&=&
-i\,\left(
\delta_{ik} F_{lj} - \delta_{il} F_{kj}
\right)
\,, \nn\\
\{A_{ij},G_{kl}\}
&=&
-i\,\left(
\delta_{ik} G_{lj} - \delta_{il} G_{kj}
\right)
\,, \nn\\
\{B_{ij},F_{kl}\}
&=&
-i\,\left(
\delta_{ik} G_{lj} - \delta_{il} G_{kj}
\right)
\,, \nn\\
\{B_{ij},G_{kl}\}
&=&
+i\,\left(
\delta_{ik} F_{lj} - \delta_{il} F_{kj}
\right)
\,, \nn
\ees
and in particular the remaining brackets vanish,
\be\label{FG0}
\{F,F\}=\{G,G\}=\{F,G\}=0.
\ee\end{subequations}
The fact that $F$ and $G$ Poisson commute will play an important role in the following.

For later convenience, we express the $\SL(2,\C)$ observables \Ref{JJ} and \Ref{ABFG} in terms of the $(t,u)$ spinors. We find:
\be
\la u_i|t_j \ra\,=\,A_{ij}-iB_{ij},
\qquad
\la t_i|u_j \ra\,=\,A_{ij}+iB_{ij},
\ee
\be
[t_i|t_j \ra\,=\,F_{ij}-iG_{ij}\,,
\qquad
[u_i|u_j \ra\,=\,F_{ij}+iG_{ij}\,,
\ee
as well as
\bes
\vJ^{L}_i\cdot\vJ^{L}_j &=&
\f14\la t_i|u_i\ra \la t_j|u_j\ra - \f12 \la t_j|t_i][ u_i|u_j\ra,
\label{JJtuL} \\
\vJ^{R}_i\cdot\vJ^{R}_j &=&
\f14\la u_i|t_i\ra \la u_j|t_j\ra - \f12 \la u_j|u_i][ t_i|t_j\ra.\label{JJtuR}
\ees

\subsection{The $\GL(N,\C)$ Action and Squashed Spinor Configurations}

Let us pause for a moment, and give some further details on the observables $A_{ij},\,B_{ij}$ and the $\GL(N,\C)$ transformations that they generate. This will be useful when treating the quantum theory in \cite{SL2Cquant}.\footnote{In the SU(2) case recalled earlier, there is a $\U(N)$ action generated by the observables $E_{ij}$, and this turns out to be very useful to understand the structure of the space of quantum intertwiners, and essential to define proper coherent intertwiner states \cite{un1,un2}.}
First, considering the Lie algebra generated by $A_{ij}$ and $B_{ij}$, we have the two obvious invariants, which commute with all other $\gl(N,\C)$ operators:
\be
\cA\equiv\tr A,\qquad
\cB\equiv\tr B\,.
\ee
Anticipating on the geometric interpretation to be discussed below in Section \ref{SecPoly}, these two invariants play the role of the total area of the polyhedron. They do not change under $\GL(N,\C)$ transformations acting on the spinors.

Computing the Poisson bracket $\{\sum_{j,k}\alpha_{jk}A_{jk}+\beta_{jk}B_{jk},\cdot \}$ with the spinors $z_i$ and $w_i$ for arbitrary Hermitian parameters $\alpha$ and $\beta$, we get the infinitesimal $\GL(N,\C)$ action. We can then exponentiate this action, to find
\bes
e^{\{\sum_{j,k}\alpha_{jk}A_{jk}+\beta_{jk}B_{jk},\cdot \}}\,|t_i\ra&=&\sum_jM_{ij}\,|t_j\ra
\qquad\textrm{with}
\quad
M=e^{i(\alpha+i\beta)}, \nn\\
e^{\{\sum_{j,k}\alpha_{jk}A_{jk}+\beta_{jk}B_{jk},\cdot \}}\,|u_i\ra&=&\sum_j\tM_{ij}\,|t_j\ra
\qquad\textrm{with}
\quad
\tM=e^{i(\alpha-i\beta)}=\,(M^\dagger)^{-1}. \nn
\ees
When the parameter $\beta$ vanishes, then $\tM=(M^\dagger)^{-1}=M$ and the $\GL(N,\C)$ transformations reduces to a unitary transformation in $\U(N)$.

We can go further in investigating the $\gl(N,\C)$ structure of the  space of the $\SL(2,\C)$-invariant observables.
The $\sl(2,\C)$ closure constraints, $\vJ=\vK=0$, can be written in terms of the spinors $t$ and $u$ as matrix equations,
\begin{subequations}\label{clostu}
\bes
\sum_i|t_i\ra\la u_i|&=&\f12\sum_i\la u_i|t_i\ra \id\,=\,\f12(\cA-i\cB)\,\id, \\
\sum_i|u_i\ra\la t_i|&=&\f12\sum_i\la t_i|u_i\ra \id\,=\,\f12(\cA+i\cB)\,\id.
\ees\end{subequations}
Then using the expression of the matrices $A$ and $B$ as scalar products between the spinors $t$ and $u$, we get quadratic constraints on the $\gl(N,\C)$ generators,
\be
(A-iB)^2\,=\,\f12(\cA-i\cB)\,(A-iB),\qquad
(A+iB)^2\,=\,\f12(\cA+i\cB)\,(A+iB).
\ee
These matrix equations have the form of a polynomial equation $C^2-\f{\tr C}2\,C=0$ for a complex matrix $C$. Excluding the cases when $\tr C=0$, which correspond to degenerate configurations, such an equation implies that $C$ is diagonalizable, with two possible eigenvalues, $0$ and $\f{\tr C}2$, the latter with degeneracy equal to 2. Applying this property to the matrices $A\pm iB$, it implies that there exists an invertible matrix $M\in\GL(N,\C)$ such that
\be
(A-iB)=\,\lambda\,M
\,\mat{cc|c}{1 & & \\ &1& \\ \hline && 0_{N-2}}\,
M^{-1},
\qquad
(A+iB)=\,\overline{\lambda}\,(M^\dagger)^{-1}
\,\mat{cc|c}{1 & & \\ &1& \\ \hline && 0_{N-2}}\,
M^\dagger\,.
\ee
This means that we can always perform a $\GL(N,\C)$ transformation that maps the two matrices $(A-iB)_{ij}=\la u_i|t_j\ra$ and $(A+iB)_{ij}=\la t_i|u_j\ra$ on such rank-two diagonal matrices.
As earlier in Section \ref{SecSquash}, these rank-two diagonal matrices represent squashed configurations of the spinors where all but two pairs of spinors vanish, that is
\be
(|z_1\ra,|w_1\ra),\qquad
(|z_2\ra=|z_1],|w_2\ra=|w_1]),\qquad
|z_{k\ge 3}\ra=|w_{k\ge 3}\ra=0.
\ee

This property of the $\GL(N,\C)$ action on sets of spinors satisfying the $\sl(2,\C)$ closure constraints means that  the space of $\SL(2,\C)$ intertwiners (for fixed values of $\cA$ and $\cB$) should carry at the quantum level a (irreducible) representation of $\GL(N,\C)$ with highest weight defined by a bivalent intertwiner and that we should be able to define coherent intertwiner states peaked on classical sets of spinors by the $\GL(N,\C)$ action on such bivalent intertwiners, just like when using the $\U(N)$ action on $\SU(2)$ intertwiners \cite{un2,simplicity1}. This will be studied in \cite{SL2Cquant}.

\section{Holomorphic Simplicity Constraints}
\label{secsimpl}

We are finally ready to discuss the main result of this paper.
We have introduced a phase space with an $\SL(2,\C)$ action on it, and a complete set of $\SL(2,\C)$-invariant quantities.
On the space of $\SL(2,\C)$ invariants, we consider the following quadratic simplicity constraints,
\be \label{simplquadra}
\vJ^{R}_i\cdot\vJ^{R}_j=e^{2i\theta} \,\vJ^{L}_i\cdot\vJ^{L}_j,\qquad\forall i,j,
\ee
where
\be
e^{i\theta(\g)}=\f{1+i\g}{1-i\g}, \qquad \g=\tan\f\theta2
\ee
is the Immirzi parameter.
In the following, we will simply denote by $\theta$ this specific function of $\g$.
Note that $\vJ^{R}_i\cdot\vJ^{R}_j$ is the complex conjugate of $\vJ^{L}_i\cdot\vJ^{L}_j$, so that they
can only differ by a phase. From this point of view, fixing the Immirzi parameter amounts to fixing that phase uniquely for all couples of edges $(i,j)$.
To be precise, these are the primary simplicity constraints. We do not discuss the secondary constraints in this paper.
For an introduction to simplicity constraints and their role in spin foam models, see \cite{tesi}. What is relevant for us, is that the constraints are second class, and furthermore their algebra does not close, a situation which requires care to define the quantum theory.
They have a non-trivial center, given by the diagonal constraints ($i=j$),
\be
(\vJ^{R}_i)^2 = e^{2i\theta}\,(\vJ^{L}_i)^2,
\ee
or equivalently by
\be
M_i = J^2_i-K^2_i +2\cot\theta \, K_i\cdot J_i = 0.
\ee
The off-diagonal ones ($i\neq j$) do not commute on the constraints surface, thus establishing their second class nature, as can be easily seen from that fact that in both right and left sectors,
\be
[\vJ_i\cdot \vJ_j, \vJ_i\cdot \vJ_k] = i \vJ_i \cdot \vJ_j \w \vJ_k.
\ee

The constraints exist also in linear form,\footnote{The quadratic constraints have two sets of solutions, however these correspond simply to the flip $\g\mapsto-1/\g$ in the linear constraints.} which states that
the special combination\footnote{Our covariant conventions are as follows:
$$
K_a\equiv J_{0a}, \qquad J_a\equiv \eps_{abc} J_{bc}, \qquad \star=\f12\eps^{IJ}_{KL}, \qquad \eps^{0123}=1.
$$
}
\linebreak $(J_i-\g\star J_i)_{IJ}$ has a unique (timelike) normal for all $i$,
\be\label{simplLin}
{\cal N}^I (J_i-\g\star J_i)_{IJ}=0.
\ee
This constraint is more familiar in the literature in the time gauge ${\cal N}^I=(1,0,0,0)$,
where it takes the form
\be\label{defC}
\vec C_i \equiv \vK_i +\gamma \vJ_i = 0,
\ee
and for each face $i$,
\be
\{ C^a, C^b \} = \eps^{abc} [2 \g C^c - (1+\g^2) J^c],
\ee
which shows explicitly that they are first class only in the self-dual theory for which $1+\g^2=0$.
The expression \Ref{defC} is not gauge invariant, thus can not be imposed as such on the phase space of
SL$(2,\C)$ invariant configurations (and equivalently on intertwiners at the quantum level).
In the following, we will concentrate on the covariant form of the constraints.

Because they are second class, the simplicity constraints should not be imposed as strong operator equations at the quantum level.
In the current spin foam models, the quantization is done imposing them weakly, as expectation values, with the exception of the center, which is treated differently and imposed as a strong operator equation.
In \cite{simplicity1,simplicity2}, it was proposed to treat all constraints on the same footing, and to realize the weak imposition as a strong imposition of the holomorphic part of the constraints. This is possible thanks to the fact that at the quantum level, the holomorphic simplicity constraints involve then only annihilation operators, and thus commute with each other.
It amounts to a weak imposition by coherent states. This is similar to what happens in the EPRL-FK models, except one will use the complex structure introduced by the spinors, and also the diagonal constraints will be imposed weakly.

This is precisely what happens in the Gupta-Bleuler prescription for the gauge-fixing condition in quantum electrodynamics, and crucially relies on a complex structure with respect to which perform the splitting.
The space that we have introduced has a natural complex structure, and this can be now exploited to achieve the splitting.
Notice in fact that the constraints \Ref{simplquadra} and \Ref{simplLin} are not holomorphic in the spinorial variables. For instance, \Ref{defC} reads
\be
(1+i\g)\, \bra{t}\vec\sigma \ket{u} = (1-i\g)\, \bra{u} \vec\sigma \ket{t},
\ee
which involves both the fundamental spinor variables and their complex conjugates.

Therefore, we want to introduce new simplicity constraints such that
\begin{itemize}
\item they are holomorphic in the $u_i,t_i$ variables;
\item they Poisson commute with each other;
\item they  imply the usual quadratic simplicity constraints (the diagonal ones as well as the cross diagonal ones, which in turn imply the linear constraints).
\end{itemize}

\subsection{Holomorphic  Simplicity Constraints}
\label{simpleL}
Following the previous work in the Euclidean case \cite{simplicity1,simplicity2} by two of us, we introduce holomorphic constraints linear in the invariants $F_{ij},G_{ij}$, or equivalently $[t_i|t_j \ra,[u_i|u_j \ra$:
\be\label{simple_def}
\cC_{ij} \equiv [t_i|t_j \ra - \alpha^2 [u_i|u_j \ra = 0, \qquad\alpha\in\C\setminus \{0\}.
\ee
This is the key equation of the paper.
$C_{ij}$ is an antisymmetric matrix, quadratic in the spinorial variables like the linear constraints.
As desired, the new constraints are fully holomorphic in $u$ and $t$,\footnote{In terms of the initial spinors \Ref{zwPB}, the constraints are holomorphic in $z$ and anti-holomorphic in $w$.}
and \emph{they commute with each other},
\be
\{\cC_{ij}, \cC_{kl}\} = 0,
\ee
as a direct consequence of \Ref{FG0}.
We now prove that they imply the quadratic \Ref{simplquadra} and linear \Ref{defC} constraints.

\begin{prop}
If the holomorphic simplicity constraints
$\cC_{ij}=0$ are satisfied for all pairs $i,j$, and the $\sl(2,\C)$-closure constraints hold, $\sum_i \vJ_i=\sum_i\vK_i=0$, then the quadratic simplicity constraints are satisfied,
\be
\vJ^{R}_i\cdot\vJ^{R}_j=e^{2i\theta}\,\vJ^{L}_i\cdot\vJ^{L}_j,\qquad\forall i,j,
\ee
with
\be
\alpha=r e^{i\f\theta2}.
\ee
\end{prop}

\begin{proof}
Let us start with $[t_i|t_j \ra= r^2e^{i\theta} [u_i|u_j \ra$ and its complex conjugate $\la t_j|t_i]= r^2e^{-i\theta} \la u_j|u_i]$. This implies that
\be
\la u_j|u_i][t_i|t_j\ra=e^{2i\theta}\la t_j|t_i][u_i|u_j\ra.
\label{terme1}
\ee
Then let us use the $\sl(2,\C)$ closure constraints \Ref{clostu} to evaluate the scalar product $\la u_i|t_i\ra$,
$$
\la u_i|t_i\ra=[t_i|u_i]
\,=\,
\f{\sum_j[t_i|t_j\ra\la u_j|u_i]}{\f12\sum_k \la u_k |t_k\ra}.
$$
Using the equality \eqref{terme1} derived above and then again the closure constraints, we get
\be
\la u_i|t_i\ra=e^{2i\theta}
\,\f{\sum_j \la t_j |u_j\ra}{\sum_k \la u_k |t_k\ra}\,
\la t_i|u_i\ra.
\ee
Summing this equation over the index $i$ gives us $({\sum_j \la u_j |t_j\ra})^2=e^{2i\theta}\,(\sum_j \la t_j |u_j\ra)^2$, that is
$$
\sum_j \la u_j |t_j\ra\,=\,\eps\, e^{i\theta}\,\sum_j \la t_j |u_j\ra, \qquad \eps=\pm.
$$
Inserting this in the previous equation gives
\be\label{utpm}
\la u_i|t_i\ra=\,\eps \, e^{i\theta} \la t_i|u_i\ra,\quad\forall i\,,
\ee
where the sign $\eps$ is the same for all $i$'s.
This implies
\be
\la u_i|t_i\ra\la u_j|t_j\ra=
e^{2i\theta} \la t_i|u_i\ra\la t_j|u_j\ra\,,
\label{terme2}
\ee
which does not depend on the sign. Finally, injecting the two equalities \eqref{terme1} and \eqref{terme2} in the expressions \eqref{JJtuL} and \eqref{JJtuR} for the scalar products $\vJ^{R,L}_i\cdot\vJ^{R,L}_j$, we obtain
$$
\vJ^{R}_i\cdot\vJ^{R}_j
\,=\,e^{2i\theta}\,\vJ^{L}_i\cdot\vJ^{L}_j,\quad\forall i,j
$$
as claimed. The modulus $|\alpha|\equiv r$ is irrelevant.

\end{proof}

The sign ambiguity above in \Ref{utpm} corresponds to two sectors of solutions of the quadratic constraints,
related in the linear constraints by the flip $\g\mapsto -1/\g$, or equivalently
$$
\theta\mapsto \theta+\pi, \qquad \alpha \mapsto i \alpha.
$$
The holomorphic constraints \Ref{simple_def} are sensitive to this flip and
distinguishes the two sectors, just as the linear simplicity
constraints. This is similar to what happens in the Euclidean case
\cite{simplicity2}.

The result in the proposition has an important immediate consequence, which we now present and which introduces a version of the constraints linear in the spinors.

\begin{lemma}
\label{cool}

Let us assume that we have two collections of spinors $t_i$ and $u_i$ satisfying $\cC_{ij}=0 \ \forall i,j$. If not all scalar products $[u_i|u_j \ra$ vanish, there exists a (unique) 2$\times$2 matrix $\Lambda\in\SL(2,\C)$ such that
\be\label{tLu}
|t_i\ra = \alpha\,\Lambda\,|u_i\ra,\quad\forall i\,.
\ee

\end{lemma}

\begin{proof}
First, we can reabsorb the constant $\alpha\in\C$ in the definition of the spinors $u_i \arr \alpha u_i$. This allows to set $\alpha$ to 1 without loss of generality. Then let us choose $i,j$ such that $[u_i|u_j \ra \ne 0$. This means in particular that $u_j$ and $u_i$ are not colinear. Hence, also $[t_i|t_j \ra \ne 0$ and the spinors $t_j$ and $t_i$ are not colinear. Thus there exists a unique invertible 2$\times$2 matrix $\Lambda$ that sends the spinors $u_i$ and $u_j$ respectively on $t_i$ and $t_j$:
$$
|t_i\ra = \Lambda\,|u_i\ra,\quad
|t_j\ra = \Lambda\,|u_j\ra\,.
$$
Due to the hypothesis that $[t_i|t_j \ra=  [u_i|u_j \ra$, this matrix has necessarily a unit determinant and thus lies in $\SL(2,\C)$. An explicit formula is given by
\be
\Lambda = \f{|t_j\ra [u_i|-|t_i\ra [u_j|}{[u_i|u_j \ra}\,.
\ee
Let us consider an arbitrary index $g$ and compute the image of $|u_k\ra$:
\be
\Lambda\,|u_k\ra
\,=\,
\f{|t_j\ra [u_i|u_k\ra-|t_i\ra [u_j|u_k\ra}{[u_i|u_j \ra}
\,=\,
\f{|t_j\ra [t_i|t_k\ra-|t_i\ra [t_j|t_k\ra}{[t_i|t_j \ra}
\,=\,
|t_k\ra\,,
\ee
which proves the desired result.
\end{proof}

If we take into account also the closure constraints, we can further show that $\Lambda$ is a pure boost (up to a flip, see below), and derive directly the linear simplicity constraints \Ref{simplLin}.
Accordingly, let us assume that the holomorphic simplicity constraints \Ref{simple_def} and the $\sl(2,\C)$ closure constraints \Ref{clostu} are satisfied.
This implies that at least one scalar product $[u_i|u_j \ra$ does not vanish,\footnote{If all scalar products $[u_i|u_j \ra$, and thus $[t_i|t_j \ra$, vanish, all spinors $u_i$ are colinear to each other, and also all spinors $t_i$ are colinear to each other, which implies that the matrix $\sum_i |t_i\ra\la u_i |$ is of rank-one and can not be proportional to the identity, which violates the closure constraints.}
thus we can apply the previous lemma.
Combining \Ref{tLu} with the closure constraint \Ref{clostu}, we get:
\be
\alpha\,\Lambda\,\sum_i |u_i\ra \la u_i|\,=\,\sum_i |t_i\ra \la u_i|
\,=\,\f12\sum_i  \la u_i|t_i\ra\,\id.
\ee
Since the matrix $\sum |u_i\ra \la u_i|$ is Hermitian, we conclude that the matrix $\Lambda$ is either Hermitian or anti-Hermitian, depending on the sign of \Ref{utpm}: $\Lambda^\dagger=\eps\, \Lambda$.
Being in $\SL(2,\C)$, this means that $\Lambda$ is either a pure boost, say $\Lambda=\exp\{\vv\cdot\vsigma\}$ with $\vv\in\R^3$, or a ``flipped boost'', $\tl\Lambda=\exp\big\{(i\f\pi2 \f\vv{|\vv|}+\vv)\cdot\vsigma\big\}$, in which the boost is combined with a rotation of $\pi$ along its axis.

In both cases, we recover the linear simplicity constraints \Ref{simplLin} with normals $\cal N$ determined by the action
of $G\equiv \Lambda^{-1/2}$ or $\tl G\equiv \tl\Lambda^{-1/2}$, on $(1,0,0,0)$.
To see this, we take $G=e^{i\vv\cdot\vK}=e^{\f12\vv\cdot(\vJ^L-\vJ^R)}=e^{\f12\vv\cdot\vJ^L}e^{-\f12\vv\cdot\vJ^R}$ and compute its action on $\vJ_i^{L,R}$. Using the transformation properties \eqref{JLaction} and \eqref{JRaction}, we find
\bes
e^{i\vv\cdot\vK} \vartriangleright \vJ^L_i &=&
e^{+\f12\vv\cdot\vJ^L}\vartriangleright \vJ^L_i
\,=\,
\f12\la t_i | e^{-\vv\cdot\f{\vsigma}2} \vsigma e^{\vv\cdot\f{\vsigma}2} | u_i\ra,
\\
e^{i\vv\cdot\vK} \vartriangleright \vJ^R_i &=&
e^{-\f12\vv\cdot\vJ^R}\vartriangleright \vJ^R_i
\,=\,
\f12\la u_i | e^{+\vv\cdot\f{\vsigma}2} \vsigma e^{-\vv\cdot\f{\vsigma}2} | t_i\ra\,.
\ees
Then, using the map $\Lambda=e^{\vv\cdot\vsigma}$ between $t_i$ and $u_i$,
it is clear that we have:
\be
e^{i\vv\cdot\vK} \vartriangleright (\vJ^R_i - e^{i\theta}\, \vJ^L_i) = 0.
\ee
Namely, the boost maps all the 3-vectors $\vJ_i^L$ on the right counterparts $\vJ_i^R$, thus recovering the linear simplicity constraints. Indeed, translating into the real $\sl(2,\C)$-generators $\vJ_i=(\vJ^L_i+\vJ^R_i)/2$   and $i\vK_i=(\vJ^L_i-\vJ^R_i)/2$, we find \Ref{defC},
\be
e^{i\vv\cdot\vK} \vartriangleright (\vK_i+\,\tan\f\theta2\, \vJ_i)\,=\,0\,.
\ee

Analogously one proves that the flipped boost solution gives
\be
\tl G \vartriangleright (\vJ^R_i + e^{i\theta} \vJ^L_i) =0.
\ee
That is, although the holomorphic simplicity constraints \Ref{simple_def} distinguish the two sectors $\theta$ and $\theta+\pi$, the reconstruction of the polyhedron mixes them again. This seems to be a peculiar feature of the Lorentzian signature, as it does not happen in the Euclidean case \cite{simplicity2}. This will be further investigate elsewhere, in particular to find the right observable to distinguish the two sectors.

The results can be summarized in the following proposition.

\begin{prop}
\label{final}
If the holomorphic simplicity constraints
$\cC_{ij}=[t_i|t_j \ra- r^2e^{i\theta} [u_i|u_j\ra=0$ and the $\sl(2,\C)$ closure constraints $\sum_i \vJ_i=\sum_i\vK_i=0$ hold, then there exist  two (unique) $SL(2,\C)$ transformations, a pure boost $G$ and a (squared root of a) flipped boost $\tl G$, such that:
\be\label{linKL}
G\vartriangleright (\vK_i\,+\,\tan\f\theta2\,\vJ_i)\,=\,0 \quad \forall i, \qquad
\tl G\vartriangleright (\vK_i\,-\,\cot\f\theta2\,\vJ_i)\,=\,0 \quad \forall i.
\ee
In covariant notations, this statement translates into the existence of a (unique) common unit (and future-oriented) time-like vector $\cN_\mu$ for all bivectors $(J_i)_{IJ}$ such that
$$
\cN^I (J_i-\gamma \star J_i)_{IJ}=0 \quad \forall i, \qquad
\tl\cN^I (J_i+\f1\gamma \star J_i)_{IJ}=0 \quad \forall i,
$$
with $\gamma=\tan\f\theta2$. $\cN$ (resp. $\tl \cN$) is obtained through the action of the boost $G$ (resp. $\tl G$) on the time-like vector at the origin $(1,0,0,0)$.

\end{prop}

The equivalence with the linear constraints \Ref{linKL} just proved, shows also that the holomorphic constraints $\cC_{ij}$ are not all independent.
In fact, there are only $2N-3$ independent constraints, because of the existence of the Pl\"ucker relations recalled earlier.\footnote{These  are the classical
counterpart to the Mandelstam relations defining the dimensionality of the intertwiner space. Our equivalence result implies that there are $(N-2)(N-3)$ independent 
Pl\"ucker relations.} 
A simple counting argument at this point is useful to highlight how in our proposal the familiar second class part of the quadratic constraints, briefly reviewed earlier, is exchanged for holomorphic constraints.
First of all, to simplify the discussion, let us perform the counting as if the quadratic constraints formed a closed algebra. Diagonal and off-diagonal simplicity constraints can be distinguished, and we have $2N$ real diagonal constraints (the usual matching of the left- and right-handed norms, plus the matching of the spinor phases -- these are a trivial addition due to the use of spinors instead of bi-vectors as fundamental variables), and the remaining $2(N-3)$ {real} off-diagonal constraints are second class. Our proposal would then be to trade the latter for $N-3$ holomorphic constraints. 
However, the true situation is a bit more complicated, because the algebra of quadratic constraints does not close, hence an infinite number of constraints arises, which are not taken into account in the above naive counting.
To solve this problem, what we are proposing is to abandon the distinction between diagonal and off-diagonal,\footnote{It should also be said that the distinction is anyway lost if one takes into account also the secondary constraints, since then all constraints do not Poisson-commute on shell, see e.g. \cite{Alexandrov}.}
and to treat all constraints on the same footing. Then, to replace the \emph{total} number of $4N-6$ real constraints by the new $2N-3$ independent holomorphic constraints.

To make the above counting more concrete, and connect with quantization, let us consider the familiar case of the tetrahedron, $N=4$. One could forget the problem that the algebra of quadratic constraints does not close, and focus on the 2 independent second class off-diagonal constraints. These can be traded for a single holomorphic constraint. In the quantum theory, it can be imposed \`a la Gupta-Bleuler, looking for coherent states which are eigenstates of the annihilation operators and thus authomatically impose the two constraints weakly. This is indeed achieved with the LS coherent intertwiners \cite{LS}, and it leads to the EPRL model \cite{EPR,LS2,FK,EPRL}. 
Conversely, one can use the \emph{closed} system of holomorphic constraints introduced in this paper. Then one is led to construct a new set of more complete coherent states, which take into account the full algebra of constraints. This is what was done in \cite{simplicity2} for the Euclidean case. The construction in the Lorentzian case will be the subject of the follow up paper \cite{SL2Cquant}.

Summarizing, we have defined  holomorphic simplicity constraints for a space of $\SL(2,\C)$ invariants which imply the ordinary quadratic and linear simplicity constraints. They are fully holomorphic in the fundamental spinorial variables, and commute with each other.
This extends to the Lorentzian case the previous results obtained in the Euclidean case \cite{simplicity1,simplicity2}.
The main differences with the Euclidean case is that we now have a boost and not a $\SU(2)$-rotation relating the 3-vectors in the self-dual and anti-self-dual sectors and that the Immirzi parameter $\gamma$ is free to run in the whole real line $\R$ without encountering some problem in $\pm 1$.

\subsection{Reduction from $\SL(2,\C)$ to $\SU(2)$-Intertwiners and the $Y$-map}\label{SecPoly}

Now that the algebraic structures are cleared up, let us look into the geometrical interpretation of our construction.
First, recall the meaning of the $\SU(2)$ case. We have $N$ spinors, say $\zeta_i$, that satisfy the closure constraint $\sum_i |\zeta_i\ra\la \zeta_i|\propto \id$. The constraints generate global $\SU(2)$ transformations on the spinors. Furthermore, these spinors determine 3-vectors $\vJ_i=\f12\la \zeta_i|\vsigma|\zeta_i\ra$ and the closure constraint then reads simply $\sum_i \vJ_i=0$. This determines a unique polyhedron embedded
in $\R^3$, which has  $N$ faces such that the vectors $\vJ_i$ are the normal vectors to each face (with the norm of the vector determining the area of the corresponding face). Details on the reconstruction of the polyhedron from the normal vectors can be found in \cite{polyhedron}.

In the $\SL(2,\C)$ case, the geometric picture is doubled. We have now a pair of spinors, $(u_i,t_i)$, for each face, corresponding to a bivector $J_i^{IJ}=(\vJ,\vK)\in\w^2\, \R^{3,1}$. The bivector represents the two-normal to the face embedded in Minkowski spacetime. Then the closure constraint $\sum_i J_i^{IJ}=0$ implies the existence of \emph{two} framed polyhedra,
corresponding to the closures of the self and antiself-dual sectors, $\sum_i \vJ^R_i=\sum_i \vJ^L_i=\vec{0}$
(or equivalently to $\sum_i \vJ_i=\sum_i\vK_i=\vec{0}$). 

The role of the simplicity constraints is then to identify the self and antiself-dual sectors
(up to a $\g$-dependent phase),
$\vJ^R_i = e^{i\theta} \vJ^L_i$.
This guarantees the existence of a unique polyhedron, determined by the bivectors all lying in the same spacelike 3d surface orthogonal to $\cN$, see proposition 2. The role of the Immirzi parameter is to determine the true area bivector as $B^{IJ}=(\id -\g \star) J^{IJ}$.

Our phase space description of polyhedra generalizes the work of Baez and Barrett \cite{Baez}, where a construction is given of the classical and quantum tetrahedron in $\R^4$, in three ways: inclusion of the Immirzi parameter, Lorentzian signature, and arbitrary valence. That is, we constructed the phase space for a classical
3d polyhedron in Minkowski spacetime.

\medskip

It is interesting to relate the SU(2) and $\SL(2,\C)$ structures at the level of an action principle for the phase spaces.
We start with the  phase space for $\SU(2)$ with the spinor variables $\zeta_i$ satisfying the closure constraint. As shown in \cite{spinor1,simplicity2}, an elegant way to represent this phase space structure is to encode it in the following action,
\be\label{Ssu2}
S[\zeta_i,\Theta]
\,\equiv\,
\int dt\,\sum_i \Big( -i\la \zeta_i|\pp_t \zeta_i\ra + \la \zeta_i|\Theta| \zeta_i\ra\Big),
\ee
where $t$ is an external time parameter with respect to which the spinors will evolve. The traceless (and Hermitian) 2$\times$2 matrix $\Theta$ is a Lagrange multiplier enforcing the closure constraint $\sum_i |\zeta_i\ra\la \zeta_i|$ generating the $\SU(2)$ gauge transformations.

We now define a map $Y$ that depends on an arbitrary boost $G\in\SL(2,\C)$, $G=G^\dagger$, and an arbitrary complex number $\beta\in\C$. We consider only the principal sector, with $\Lambda$ and thus $G$ pure boosts.
The 2$\times$2 matrix $G$ is of the type $\exp(\vec{x}\cdot\vsigma)$ with $\vec{x}\in\R^3$.
Our $Y$-map sends sets of spinors $\{\zeta_i\}$ satisfying the closure constraints to sets of pairs of spinors $\{(t_i,u_i)\}$ satisfying the $\sl(2,\C)$ closure constraints and the holomorphic simplicity constraints:
\be
Y_{G,\beta}: \{\zeta_i\}
\quad
\longmapsto
\quad
\bigg{\{}
|t_i\ra=\beta G\,|\zeta_i\ra
\,,\,
|u_i\ra=\beta^{-1}G^{-1}\,|\zeta_i\ra
\bigg{\}}
\ee
We first check the $\sl(2,\C)$ closure constraints using the fact that $G$ is Hermitian (by hypothesis, since it is a pure boost):
$$
\sum_i |u_i\ra\la t_i|
\,=\,
\f{\bbeta}{\beta}G^{-1}\,\sum_i |\zeta_i\ra\la \zeta_i|\,G^\dagger
\,=\,
\f{\bbeta}{\beta}\,\left(\f12\sum_i \la \zeta_i|\zeta_i\ra\right)\,G^{-1}\,G
\,\propto\,\id\,.
$$
We also easily check that the holomorphic simplicity constraints hold using that $G\in\SL(2,\C)$:
\be
\left\{
\begin{array}{l}
{[}t_i|t_j\ra \,=\, {\beta}^{2}  {[}\zeta_i|\zeta_j\ra {} \\
{[}u_i|u_j\ra \,=\, \beta^{-2}{} [\zeta_i|\zeta_j\ra
\end{array}
\right.
\quad\Rightarrow\quad
[t_i|t_j\ra\,=\,\beta^4\,[u_i|u_j\ra,
\ee
where $\alpha=\beta^2$ according to the definition \eqref{simple_def} of our holomorphic simplicity constraints.

Then we check that the phase space structure and the Poisson brackets between the $t$'s and the $u$'s is the appropriate one. In order to do so, we compute the kinematical term of our action:
\be
\la u_i|\pp_t t_i\ra+\la t_i|\pp_t u_i\ra
\,=\,
\left(\f{\beta}{\bbeta}+\f{\bbeta}{\beta}\right)\,\la \zeta_i|\pp_t \zeta_i\ra
\,+\,
\f{\beta}{\bbeta}\la \zeta_i| G^{-1}\pp_t G | \zeta_i\ra
\,+\,
\f{\bbeta}{\beta}\la \zeta_i| G\pp_t G^{-1} | \zeta_i\ra\,.
\ee
So it seems that we have to get rid of the term with derivative of the $G$-matrix. However, we now sum over the index $i$ and use the closure constraint $\sum_i |\zeta_i\ra\la \zeta_i|\propto\id$, then the terms in $G$ turn to be proportional to $\tr G^{-1}\pp_t G$, which vanishes since $\det G =1$. Hence,
\bes
S[\zeta_i,\Theta] &=& \int dt\, \sum_i \Big(-i\la \zeta_i|\pp_t \zeta_i\ra + \la \zeta_i|\Theta| \zeta_i\ra\Big) = \nn\\\nn
&=& \f{1}{\cos\f\theta 2}\, \int dt\, \sum_i \Big(-\f{i}2(\la u_i|\pp_t t_i\ra+\la t_i|\pp_t u_i\ra)
+ \la t_i|\Theta| u_i\ra + \la u_i|\tTheta| t_i\ra \Big) \,-\,\sum_{i,j}  \phi_{ij}\,\cC_{ij} \\
&=& S[t_i,u_i,\Theta,\tTheta,\phi],
\ees
where $\f\theta 2$ is the argument of $\alpha$ as before and the Lagrange multipliers $\phi$'s are anti-symmetric under $i\leftrightarrow j$. This equality holds due the lemma \ref{cool} which shows that the holomorphic  constraints on $t_i$ and $u_i$ do imply the existence of a boost $G$ relating the $t_i$'s to the $u_i$'s and thus such that the $Y$-map is invertible, $Y_{G,\beta}(\zeta_i)=(t_i,u_i)$.

This concludes the proof that the spinor phase space for $\SL(2,\C)$ with the holomorphic  constraints is isomorphic to the spinor phase space for $\SU(2)$. In the context of spin foam models, this is a desired feature: the Hilbert space of boundary states for a Lorentzian spin foam models with gauge group $\SL(2,\C)$ and satisfying the holomorphic simplicity constraints should be isomorphic to the Hilbert space of $\SU(2)$ spin networks. This is an important feature of the EPRL-FK models \cite{lift},
and it has been shown to apply also for the holomorphic constraints in the Euclidean case \cite{simplicity2}.

\section{Conclusions and Outlook}

We have introduced a phase space equipped with a complex structure and
a Hamiltonian action of $N$ copies of $\SL(2,\C)$.
It is given by $N$ twistors, or pairs of spinors, with canonical Poisson brackets, and a representation of the algebra generators quadratic in the spinors.
We computed the Casimirs of the algebra and showed that they can take arbitrary real values, thus the space is suitable to be quantized and carry the (infinite-dimensional) unitary representations of $\SL(2,\C)$, that are relevant in spin foam models.

Within this space, we have identified a complete set of global $\SL(2,\C)$ invariants as scalar products among the spinors, given by \Ref{ABFG}, and computed their Poisson algebra, given in \Ref{GLNC}.
Notably, the algebra \emph{closes}, unlike the one of the usual scalar products among generators, and contains a $\gl(N,\C)$ subalgebra.

Our work generalizes to $\SL(2,\C)$ the construction of the SU(2) algebra on twistor space appeared in \cite{twisted2,spinor2}, and the closed algebra of invariants appeared in \cite{un0,un1,spinor1}.

Next, we have defined the simplicity constraints, both in the quadratic and linear versions, exposing the fact that they are not holomorphic with respect to the spinorial variables.
The complex structure of the spinorial representation, and the control over a closed algebra of invariants, allow us to define a new set of simplicity constraints. While equivalent to the standard quadratic and linear constraints, the new ones have two key advantages: they are holomorphic, and they Poisson commute with each other. This is the main result of the paper.
Hence, the new holomorphic constraints are much more amenable to quantization than the usual ones. 
In particular, one can seek a weak imposition in terms of coherent states, similarly to what is done in the EPRL-FK models \cite{LS,LS2,FK}, but this time one will use explicitly the complex structure introduced by the spinors, and also the diagonal constraints will be imposed weakly.
This will be the subject of the follow-up work \cite{SL2Cquant}.

Our work extends to the Lorentzian case the previous results obtained in the Euclidean case \cite{simplicity1,simplicity2},
with the natural difference that we now have a boost and not a $\SU(2)$-rotation relating the 3-vectors in the self-dual and anti-self-dual sectors.
We also found a peculiar difference, with respect to the Euclidean case, in the mixing of the two sectors of solutions.

The results presented here provide a new light on the relation between simplicity constraints and polyhedra in Minkowski spacetime,
via spinors and twisted geometries. See also \cite{Wolfgang} and \cite{Johannes} for more details on covariant twisted geometries and twistors.

\appendix

\section{$\so(3,2)$ Algebra from a single Spinor}\label{AppA}

Starting from a single spinor $z$, we consider all the sets of real quadratic combinations of its components:
\bes
D &=& \f12\la z | z\ra\,, \\
\vJ(z)&=&\f12\la z|\vsigma|z\ra\,,\\
\vK(z)&=&\f14\Big{(}\la z|\vsigma|z]+[z|\vsigma|z\ra\Big{)}\,,\\
\vec M(z)&=&\f i4\left(\la z|\vsigma|z]-[z|\vsigma|z\ra\right)\,.
\ees
Together they form a closed algebra under Poisson bracket:
\be\nn
\{J_a(z),J_b(z)\}=\epsilon_{abc}J_c(z),\qquad
\{J_a(z),K_b(z)\}=\epsilon_{abc}K_c(z),\qquad
\{J_a(z),M_b(z)\}=\epsilon_{abc}M_c(z),
\ee
\be\nn
\{K_a(z),K_b(z)\}=-\epsilon_{abc}J_c(z),\qquad
\{M_a(z),M_b(z)\}=-\epsilon_{abc}J_c(z),\qquad
\{K_a(z),M_b(z)\}=\delta_{ab}D,
\ee
\be\nn
\{D,J_a(z)\}=0,\qquad
\{D,K_a(z)\}=-M_a(z),\qquad
\{D,M_a(z)\}=+K_a(z).
\ee
This algebra can be derived from the following Poisson brackets,
$$
\{\f12\la z|\sigma_a|z\ra,\la z|\sigma_b|z]\}
\,=\,
\f{-2i}2\la z|\sigma_a\sigma_b|z]
\,=\,
\f{-2i}2\la z|\delta_{ab}\id+i\epsilon_{abc}\sigma_c|z]
\,=\,
\epsilon_{abc}\la z|\sigma_c|z],
$$
as well as
$$
\{\f12\la z|\sigma_a|z\ra,[ z|\sigma_b|z\ra\}
\,=\,
\epsilon_{abc}[ z|\sigma_c|z\ra\,,\qquad
\{\la z|\sigma_a|z],[ z|\sigma_b|z\ra\}
\,=\,
4i\delta_{ab}\la z|z\ra-4\epsilon_{abc}\la z|\sigma_c|z\ra.
$$

The algebra generated by $\vJ,\vK,\vec M$ and $D$ can be recognized as an $\so(3,2)$ Lie algebra.
In covariant notation, we define the De Sitter generators $\cJ_{\mu\nu}$ on the 5-dimensional Minkowski space-time with signature (-+++-), via
\be\nn
J_a=\f12\epsilon_{abc}\cJ_{bc},\quad
K_a=\cJ_{0a},\quad
L_a=\cJ_{4a},\quad
D=\cJ_{04}\,.
\ee

Finally, notice that the three vector generators satisfy constraints:
\be\nn
\vJ^2=\vK^2=\vL^2=D^2,\qquad
\vJ\cdot\vK=\vJ\cdot\vL=\vK\cdot\vL=0\,,
\ee
which simply expresses that $(\vJ,\vK,\vL)$ form an orthonormal basis for 3-vectors in $\R^3$.

\section{Decomposition of a Bivector in Null Components} \label{AppB}

We start with an arbitrary bivector $B=(\vec{a},\vec{b})\in\R^6=\w^2\R^{3,1}$ and we want to write it as a sum of two null bivectors i.e:
\be
(\va,\vb)=(\vc,\vd)+(\ve,\vf) \qquad\textrm{with}\quad \vc\cdot\vd=\ve\cdot\vf,
\quad |c|=|d|,\quad |e|=|f|.
\ee
The original bivector is defined by 6 independent variables. The pair of null bivectors are determined by $12$ variables related by 4 constraints which reduce the number of degrees of freedom to 8. Thus we expect 2 free degrees of freedom in this decomposition, which would correspond to a free direction of one vector in $\R^3$.

Assuming that $\va$ and $\vb$ are not collinear, then $\va\w\vb\ne0$, and we take the following ansatz:
\be
\vc=\alpha\va+\beta\vb,\qquad \vd=\gamma \va\w\vb.
\ee
which is such that $\vc\cdot\vd=0$. The three remaining constraints are easily translated into equations on the coefficients $\alpha,\beta,\gamma$:
\bes
\alpha|a|^2+\beta \va\cdot\vb = \f{|a|^2-|b|^2}2, \qquad
\alpha\va\cdot\vb +\beta |b|^2 = \va\cdot\vb, \qquad
\gamma^2 |\va\w\vb|^2=|\alpha\va+\beta\vb|^2.\nn
\ees
The system is easily solved and there exists a unique solution for the three parameters $\alpha$, $\beta$ and $\gamma$. The extra freedom has actually been fixed when we chose  the vector $\vd$ to be parallel to $\va\w\vb$.

In the special case when $\va\w\vb=0$, we follow the same logic. Since $\va$ and $\vb$ are collinear, we write $\vb=\lambda\va$. Then we fix the direction of the vector $\vc$ (such that it is not collinear to $\va$) and determine its norm $|c|$ and the vector $\vd$.  We want to impose that:
\be
|c|=|d|,\quad \vc\cdot\vd=0,\quad |\va-\vc|=|\lambda\va-\vd|,\quad (\va-\vc)\cdot(\lambda\va-\vd)=0.
\ee
The two last constraints lead to:
\be
\va\cdot\vc =\f12|a|^2,\qquad \va\cdot\vd=\f\lambda 2|a|^2.
\ee
Once the direction of $\vc$ is fixed,
the equality on the scalar product $\va\cdot\vc$ actually fixes the norm of $|c|$. The three remaining constraints,  $|c|=|d|$, $\vc\cdot\vd=0$ and the equality on $\va\cdot\vd$ fix completely the vector $\vd$. The most straightforward way to see this is to decompose $\vd$ on the basis formed by $\va$, $\vc$ and $\va\w\vc$.


\end{document}